\documentclass[conference]{IEEEtran}
\IEEEoverridecommandlockouts
\usepackage{cite}
\usepackage{cuted}
\usepackage{amsthm}
\usepackage[mathscr]{euscript}
\usepackage{amsmath,amssymb,amsfonts}
\usepackage{algorithmic}
\usepackage{graphicx}
\usepackage{bm}
\usepackage{textcomp}
\usepackage{xcolor}
\newtheorem{theorem}{Theorem}

\def\BibTeX{{\rm B\kern-.05em{\sc i\kern-.025em b}\kern-.08em
    T\kern-.1667em\lower.7ex\hbox{E}\kern-.125emX}}
\begin{document}

\title{Extra DoF of Near-Field Holographic MIMO Communications Leveraging Evanescent Waves}

\author{Ran Ji, Shuo Chen, Chongwen~Huang, Wei E. I. Sha,~\IEEEmembership{Senior,~IEEE},  Zhaoyang Zhang, Chau~Yuen,~\IEEEmembership{Fellow,~IEEE}, \\
and M\'{e}rouane~Debbah,~\IEEEmembership{Fellow,~IEEE},

\thanks{R. J, S. Chen, C. Huang, W. E. I. Sha, Z. Yang, are with the Zhejiang University, China; C. Yuen is with the Singapore University of Technology and Design, Singapore; M. Debbah is M. Debbah is  with the Technology Innovation Institute,  Abu Dhabi, United Arab Emirates.}
}

\maketitle
\begin{abstract}
In this letter, we consider transceivers with spatially-constrained antenna apertures of rectangular symmetry, and aim to improve of spatial degrees of freedom (DoF) and channel capacity leveraging evanescent waves for information transmission in near-field scenarios based on the Fourier plane-wave series expansion. The treatment is limited to an isotropic scattering environment but can be extended to the non-isotropic case through the linear-system theoretic interpretation of plane-wave propagation. Numerical results show that evanescent waves have the significant potential to provide additional DoF and capacity in the near-field region.
\end{abstract}


\section{INTRODUCTION}
In MIMO (multiple-input multiple-output) systems, channel capacity grows linearly with the number of spatial degrees of freedom (DoF), which is relative to both the scattering environment and antenna array geometries. Recently, the concept of \textit{holographic MIMO} has drawn increasing attention and is regarded as one of the possible technologies used in 6G communications\cite{b1}. A holographic MIMO array consists of a massive (possible infinite) number of antennas integrating into a compact space, which can be seen as a ultimate form of spatially-constrained MIMO system and be modeled as a spatially-continuous electromagnetic (EM) apertures with asymptotically infinite antennas. Thus, a fundamental question of holographic communications is to obtain the number of spatial DoF and capacity of a holographic MIMO system. To answer this question, many previous works have investigated the continuous-space channel modelling under different propagation conditions and array geometries, e.g.,\cite{b2}-\cite{b4}. However, these models are all based on the far-field communication assumption, which might be not applicable with the dramatic increase of antenna aperture and higher operation frequencies that makes near-field communications be the practical scenario. Thus, near-field modeling, the spatial DoF and its capacity boundary are needed. There have been some works trying to fill the gap from the Green's function and electromagnetic theory, such as \cite{b5}, \cite{b6} and \cite{b7}. In these works, the authors derive channel models and estimate near-field communication performances from Green's function, corresponding to a relatively simple communication scenario. However, to the best of the authors' knowledge, near-field analysis based on channel model corresponding to complex scattering environment have not been investigated.

In this letter, we consider transceivers with spatially-constrained antenna apertures of rectangular symmetry based on spatially-stationary random monochromatic scattering propagation environment. Unlike work \cite{b8} that a Fourier orthonormal expansion was derived to compute the DoF limit in the far-field region, we focus on the near-field communication scenario and compute the DoF and capacity gain by considering evanescent waves. To the best of our knowledge, this is the initial work to analyze the DoF and capacity limit of the near-field communications under a generalized stochastic spatially-stationary monochromatic channel. Moreover, we also demonstrate our proposition through numerical simulations and verify its coherence with traditional results, i.e., the maximum DoF and capacity of the far-field region that are provided in \cite{b8} and \cite{b9}. Specifically, we take full potential of evanescent waves in near-field holographic communications, and drive its extra DoF and capacity based on Fourier plane-wave series expansion. Numerical results show that the DoF and capacity of near-field communications can be improved by 30\% respectively comparing with the traditional far-field scenario.

\section{SYSTEM MODEL}
We study a holographic communication system where the transmitter and the receiver are equipped with planar holographic MIMO surface. Considering the electromagnetic waves propagation in every direction (i.e., isotropic propagation) through a homogeneous, isotropic, and infinite random scattered medium, the 3D small-scale fading channel can be modeled as a space-frequency scalar random field \cite{b9},\cite{b10}, which is given as follows:
\begin{equation}
    {h_\omega(x,y,z):(x,y,z)\in\mathbb{R}^3,\omega\in(-\infty,\infty)}
\end{equation}
which is a function of frequency and spatial position $(x,y,z)$. Since we only consider monochromatic waves in this letter, the variant $\omega$ can thus be omitted. According to \cite{b10}, $h(x,y,z)$ can be modeled as a zero-mean, spatially-stationary and Gaussian random field.
This channel can be completely described according to \cite{b10}:
\begin{equation}
    c_h(x,y,z)=\mathbb{E} \{h^*(x',y',z')h(x+x',y+y',z+z')\}
\end{equation}
\begin{equation}
    S_h(k_x,k_y,k_z)=\iiint c_h(x,y,z)e^{-j(k_xx+k_yy+k_zz)}dxdy
\end{equation}
where $c_h(x,y,z)$ is the spatial auto-correlation function in spatial domain and $S_h(k_x,k_y,k_z)$ is the power spectral density in the wave-number domain.

\subsection{Fourier Plane-Wave Spectral Representation}
In a source-free environment, the EM nature of the small-scale fading requires realizations of $h(x,y,z)$ to satisfy the scalar Helmholtz equation in the frequency domain, which can be seen as a constraint of $h(x,y,z)$:$(\nabla^2+\kappa^2)h(x,y,z)=0$ where $\kappa=\frac{2\pi}{\lambda}$ is the module of the vector wave-number and $\lambda$ is the wavelength.
Following the derivation in \cite{b10}\cite{b11}, we can find out that the channel's power spectral density in the wave-number domain has the following form under the Helmholtz equation constraint:
\begin{equation}
    S_h(k_x,k_y,k_z)=\frac{4\pi^2}{\kappa}\delta(k_x^2+k^2_y+k^2_z-\kappa^2)
    \label{3Dpowerspectral}
\end{equation}
which is an impulsive function with wave-number support on the surface of a sphere of radius $\kappa$. $h(x,y,z)$ is decomposed into the sum of two random fields:
\begin{equation}
    h(x,y,z)=h_+(x,y,z)+h_-(x,y,z)
\end{equation}
which is further defined as the \textit{Fourier plane-wave spectral representation:}
\begin{equation}
\begin{aligned}
    h_\pm(x,y,z)&=\frac{1}{4\pi\sqrt{\pi}}\iint \sqrt{S_h(k_x,k_y)}W^\pm(k_x,k_y) \\
    &\times e^{j(k_xx+k_yy\pm\gamma(\kappa_x,\kappa_y)z)}
    \label{fourier-basedrepresentation}
    \end{aligned}
\end{equation}
where $W^+(k_x,k_y)$ and $W^-(k_x,k_y)$ are two 2D independent, zero-mean, complex-valued, white-noise Gaussian random fields and $S_h(k_x,k_y)$ is the 2D power spectral density of $h\pm(x,y,z=0)$. Conventionally, $S_h(k_x,k_y)$ is defined over a compact support $k_x^2+k^2_y\leq\kappa^2$ given by a disk of radius $\kappa$ centered on the origin (excluding a purely imaginary $\gamma(k_x,k_y)=\sqrt{\kappa^2-k^2_x-k^2_y}$ that is called as the evanescent waves because they do not contribute to far-field propagation). This limits the bandwidth of conventional $h(x,y,z)$ (in the wave-number domain) to $\pi\kappa^2$. However, as we will demonstrate in Section \uppercase\expandafter{\romannumeral4}, the far-field communication assumption might be not applicable in some short range communications scenarios because of the increase of antenna aperture and operating frequencies in the next generation communications. As a consequence, it is possible to leverage evanescent EM waves to transmit information in near-field communications.

\subsection{4D Fourier Plane-Wave Series Expansion}

In \cite{b9}, authors propose 4D Fourier plane-wave representation and 4D Fourier plane-wave series expansion of EM channels. Compared with the 2D representation, the 4D version is derived from Helmholtz equation and Green's function as well. Specifically, 4D Fourier plane-wave representation can be written as
\begin{equation}
\begin{aligned}
    h(\textbf{r},\textbf{s}) = &\frac{1}{(2\pi)^2}\iiiint a_r(\textbf{k},\textbf{r})H_a(k_x,k_y,\kappa_x,\kappa_y)a_s(\boldsymbol{\kappa},\textbf{s})\\
    &\times dk_xdk_yd\kappa_xd\kappa_y
    \label{4D1}
\end{aligned}
\end{equation}
where $a_r(\textbf{k},\textbf{r})H_a(k_x,k_y,\kappa_x,\kappa_y)=e^{j\textbf{k}^T\textbf{r}}$, $a_s(\boldsymbol{\kappa},\textbf{s})=e^{-j\boldsymbol{\kappa}^T\textbf{s}}$ respectively. (\ref{4D1}) can be approximated by 4D Fourier plane-wave series expansion:
\begin{equation}
    h(\textbf{r},\textbf{s}) \approx \sum_{}\sum_{}a_r(l_x,l_y,\textbf{r})H_a(l_x,l_y,m_x,m_y)a_s(l_x,l_y,\textbf{s})
\end{equation}
where $a_r(l_x,l_y,\textbf{r}) = e^{j(\frac{2\pi}{L_{R,x}}l_xr_x+\frac{2\pi}{L_{R,y}}l_yr_y+\gamma_r(l_x,l_y)r_z)}$ and $a_s(m_x,m_y,\textbf{s}) = e^{-j(\frac{2\pi}{L_{S,x}}m_xs_x+\frac{2\pi}{L_{S,y}}m_ys_y+\gamma_r(m_x,m_y)s_z)}$.

The channel matrix of Holographic MIMO communications can be written as
\begin{equation}
    \textbf{H} = \boldsymbol{\Phi}_r\Tilde{\textbf{H}}\boldsymbol{\Phi}_s^H = \boldsymbol{\Phi}_r e^{j\boldsymbol{\Gamma}_r} \textbf{H}_a e^{-j\boldsymbol{\Gamma}_s} \boldsymbol{\Phi}_s^H
\end{equation}
where $\boldsymbol{\Gamma}_r = diag(\gamma_r)r_z$ and $\boldsymbol{\Gamma}_s = diag(\gamma_s)s_z$. $\boldsymbol{\Phi}_r$ and $\boldsymbol{\Phi}_s$ are 2D spatial-frequency Fourier harmonics matrix. In the far-field communication scenario, the angular matrix $\textbf{H}_a$ is $\textit{semi-unitarily equivalent}$ to $\textbf{H}$ and statistically equivalent to $\Tilde{\textbf{H}}$.

%

\section{DOF AND CAPACITY OF NEAR-FIELD HOLOGRAPHIC MIMO COMMUNICATIONS}
In electromagnetics, an evanescent field, or evanescent wave, is an oscillating electric and/or magnetic field that does not propagate as an EM wave but whose energy is spatially concentrated in the vicinity of the source. As we can see that the DoF of holographic MIMO channels in far-field communication scenarios has been computed in \cite{b8}, but evanescent waves are not considered since they attenuate exponentially with the distance. When we go to beyond 5G or 6G communications with holographic MIMO  surfaces under higher operating frequencies i.e., millimeter wave or terahertz, the original far-field communication range will become in the near-field range according to the Raleigh distance $\frac{2D^2}{\lambda}$, where $D$ is the maximum linear dimension of the antenna, and $\lambda$ is the wavelength of the EM waves. One typical scenario is that the whole wall of an indoor room is equipped with holographic MIMO surfaces as the transmitter using sufficiently high carrier frequencies for communications, which might make the whole indoor environment fall in  transmitter's near-field region. For example, using a carrier frequency $f=5$ GHz (i.e., $\lambda=6$cm) and holographic MIMO surface aperture length $L_x=L_y=1$m, this roughly provide a near-field range of  $2D^2/\lambda\approx33.3$m.
\par Using the holographic MIMO surfaces to control the EM wave, it usually makes the reflected wave off the surface at an angle greater than the critical angle where evanescent waves are formed. This only happens in the near field since  the intensity of evanescent waves decays exponentially with the distance from the interface at which they are formed. Therefore, it is natural to consider of using evanescent wave to carry information besides the normal plane waves that can be applied for information transmission in the original far-field range. In this paper, we compute the limit of the average number of spatial DoF and capacity with evanescent waves. Furthermore, we will show that evanescent waves can bring the considerable DoF and capacity gain in reactive near field region from both theoretical and numerical perspectives.

\subsection{DEGREES OF FREEDOM}
\subsubsection{Far-field Communications}\label{AA}
It is well known that a band limited orthonormal series expansion has a countably-finite number of coefficients, whose cardinality determines the space dimension \cite{b8}, i.e., the available DoF. In the following, we take planar arrays and volumetric arrays as examples to derive their DoFs.

\begin{itemize}
    \item Planar arrays
\end{itemize}
\par Assuming $h(x,y,z)$ is observed over a 2D rectangle, $\mathcal{V}_2$ of side lengths $L_1 > L_2$. According to 2D Fourier Plane-Wave Series Expansion, $h(x,y) = h(x,y,0)$ is given
\begin{equation}
h(x,y) \approx \mathop{\sum\sum}\limits_{(l,m) \in \mathcal{E}}c_{l,m}\varphi_{l,m}(x,y) \quad (x,y) \in \mathcal{V}_2\label{dof_fourier_2Dexpansion}
\end{equation}
where $\varphi_{l,m}(x,y) = \varphi_l(x)\varphi_m(y)$ is the 2D Fourier basis, and $c_{l,m} = \sqrt{L_xL_y}H_{lm}(0)$. The average number of DoF is limited by the cardinality of $\{H_{lm}\}$, which is the number of lattice points falling into the 2D \textbf{inner} lattice ellipse shown in Fig.~\ref{ellipse}, that is its Lebesgue measure of $\mathcal{E}$. This yields
\begin{equation}
    \eta_2=\frac{\pi}{\lambda^2}L_xL_y\label{dof_2D}
\end{equation}

\begin{itemize}
    \item Volumetric arrays
\end{itemize}
\par Assuming $h(x,y,z)$ is observed over a 3D parallelepiped, $\mathcal{V}_3$ of side lengths $L_x,L_y$ and $L_z < \min(L_x,L_y)$. The 2D Fourier plane-wave series expansion of $h(x,y,z)$ for a fixed $z$ is
\begin{equation}
h(x,y) \approx \mathop{\sum\sum}\limits_{(l,m) \in \mathcal{E}}c_{l,m}(z)\varphi_{l,m}(x,y) \quad (x,y), \in \mathcal{V}_2\label{dof_fourier_3Dexpansion}
\end{equation}
where $c_{l,m} = \sqrt{L_xL_y}H_{lm}(z)$ with $H_{lm}(z) = H_{lm}^{+}e^{j\gamma_{lm}z}$ $+H_{lm}^{-}e^{-j\gamma{lm}z}$. Giving the pair $(l,m)$ and fixed $z$, these two functions are completely known and do not carry any information. Hence, we expect that the number of DoF over a 3D volume does not scale proportionally to $L_z$, as it happens for $L_x$ and $L_y$. By working with a single random vector, the 2D random vector field $\bold{h}(x,y) = h(x,y,\bold{z})$ is given by
\begin{equation}
\bold{h}(x,y) \approx \mathop{\sum\sum}\limits_{(l,m) \in \mathcal{E}} \bold{c}_{l,m} \varphi_{l,m}(x,y)
\end{equation}
where $\bold{c}_{lm}=\bold{\Gamma}_{lm}\bold{a}_{lm} \in \mathbb{C}^{N_z\times2}$ is a random vector of statistically-independent elements with
\begin{equation}
    \bold{\Gamma}_{lm} = [e^{j\gamma{lm}\bold{z}},e^{-j\gamma{lm}\bold{z}}] \in \mathbb{C}^{N_z\times2}
\end{equation}
and $\bold{a}_{lm} = \sqrt{L_xL_y}[H_{lm}^{+},H_{lm}^{-}]^{T} \in \mathbb{C}^{N_z\times2}$. The number of DoF are obtained by the product between (\ref{dof_2D}) and the rank of $\Gamma_{lm}$. The latter is 2 since the two columns of $\Gamma_{lm}$ are linearly independent regardless of the choice of $z$. Thus, the DoF is

\begin{equation}
    \eta_3=\frac{2\pi}{\lambda^2}L_xL_y
\end{equation}

\subsubsection{Near-field Communications}
\begin{figure}[htbp]
\centering
\includegraphics[width=0.40\textwidth]{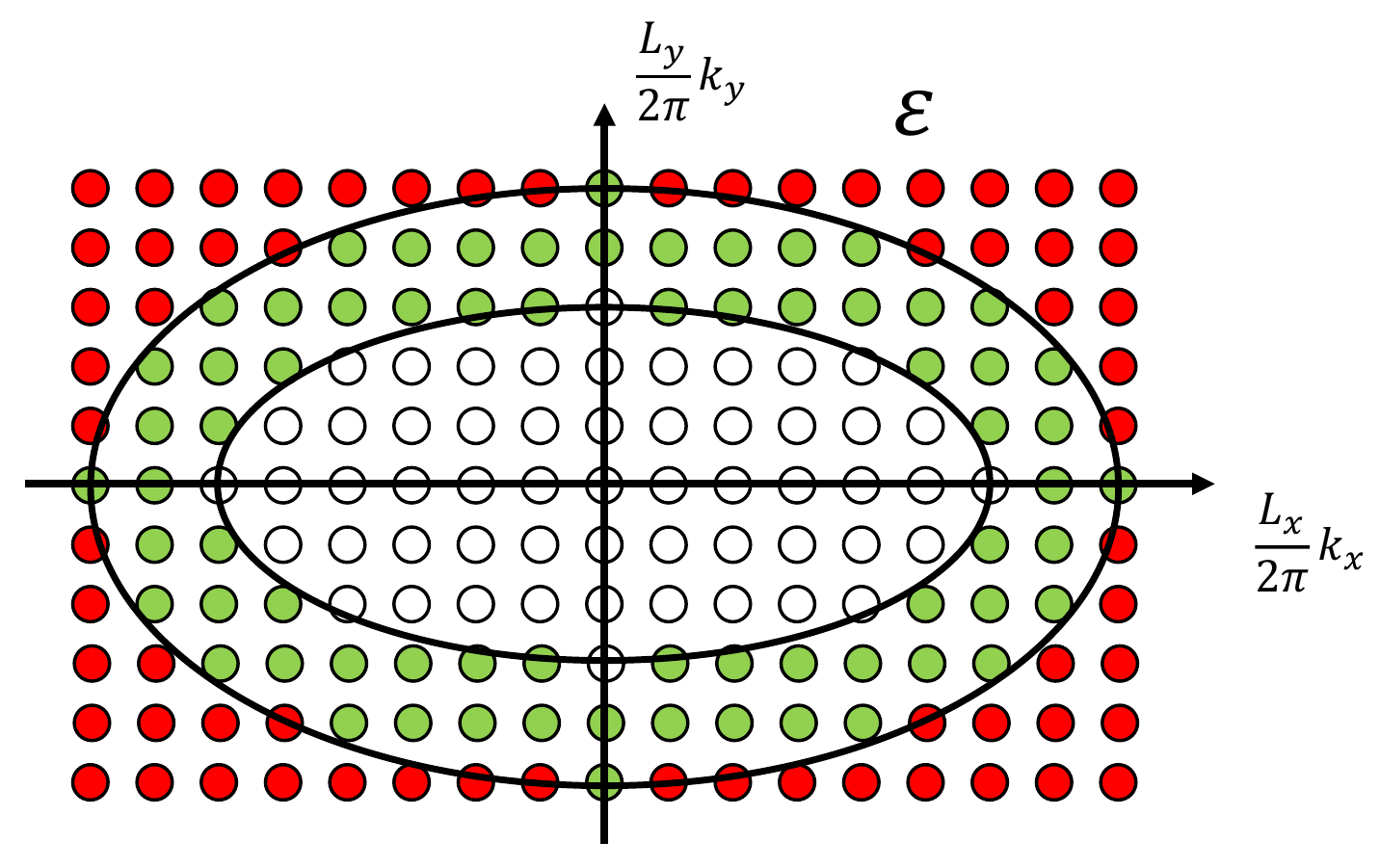}
\caption{The 2D lattice ellipse $\mathcal{E}$ wavenumber spectral support of $h(x,y,z)$}
\label{ellipse}
\end{figure}\vspace{-0mm}

We now consider evanescent waves for near-field communications and compute the number of DoF over volumetric($\mathcal{V}_3$) aperture spaces.

\newtheorem{definition}{Definition}
\begin{definition}
The amplitude of evanescent wave attenuates exponentially with propagation distance and its power can be written as
\begin{equation}
    \frac{\mathcal{P}_{receive}}{\mathcal{P}_{send}} = e^{-2k_zz}
\end{equation}
\textit{where $k_z$ is the modulus of imaginary wavenumber on z axis.}
\end{definition}

\begin{theorem}
\textit{The additional DoF can be obtained in the near-field communications that is generally given as follow, which vanishes as the distance between the transmitter and receiver planes increases.}
\begin{equation}
    \!\!DoF_{evanescent}\!=\!DoF_{far-field}*\frac{1}{(4z\pi)^2}*\lambda^2*\ln^2{\frac{\mathcal{P}_{send}}{\mathcal{P}_{noise}}}
    \label{theorem}
\end{equation}
\end{theorem}

\begin{proof}
We assume that the impact of evanescent waves is only considered when the received signal power is larger than that of the channel noise. From (16), we can obtain that the largest $k_z$ satisfies the following equation:
\begin{equation*}
    \frac{P_{noise}}{P_{send}} = e^{-2k_zz}
\end{equation*}
then we can calculate $k_z$ as
\begin{equation*}
     2k_zz = ln(\frac{P_{send}}{P_{noise}})
\end{equation*}
\begin{equation}
     k_z = \frac{1}{2z}*ln(\frac{P_{send}}{P_{noise}})
     \label{kz}
\end{equation}
When the EM wave is evanescent, its wavenumbers on three axes satisfy the following equation:
\begin{equation*}
    \kappa^2 = \kappa_x^2 + \kappa_y^2 - \kappa_z^2
\end{equation*}

Assuming the outer ellipse in Fig.~\ref{ellipse} is defined by
\begin{equation*}
    t^2 = \kappa_x^2 + \kappa_y^2
\end{equation*}

Then the additional DoF can be calculated by
\begin{equation*}
    DoF_{evanescent}=DoF_{far-field}*(\frac{{t}^2}{\kappa^2}-1)
\end{equation*}
\begin{equation}
\begin{aligned}
    DoF_{evanescent}&=DoF_{far-field}*(\frac{{\kappa}^2+\kappa_z^2}{\kappa^2}-1)\\
    &=DoF_{far-field}*\frac{{k_z}^2}{\kappa^2}\\
\end{aligned}
\label{DoFimprove}
\end{equation}
Substituting the corresponding $k_z$ in (\ref{kz}) into equation (\ref{DoFimprove}) with $\kappa=2\pi/\lambda $, it gives us the result (\ref{theorem}).
\end{proof}

\subsection{CHANNEL CAPACITY EVALUATION}
In this section, we analyze channel capacity by  considering evanescent waves. Channel capacity in far-field communication scenarios can be given by \cite{b9}:
\begin{equation}
     C = \rm \mathop{\max}_{ \textbf{Q}_a:tr(\textbf{Q}_a)<=1} \mathbb{E}\{log_2det( \textbf{I}_{n_r}+snr \textbf{H}_a \textbf{Q}_a \textbf{H}^H_a)\}
     \label{traditional_capacity}
\end{equation}
where $\textbf{Q}_a = \mathbb{E}\{\textbf{x}_a^{\rm H} \textbf{x}_a \}$, and $\textbf{I}_{n_r}$ is an identity matrix of dimension $n_r$. $\textbf{H}_a \in \mathbb{C}^{n_r \times n_s}$ is the angular response matrix describing the channel coupling between source and receive propagation directions, assuming that there are $n_s$ and $n_r$ sampling points on the transmitting and receiving side, respectively. However, Eq. (\ref{traditional_capacity}) is based on the assumption that $\bm{\widetilde{{\rm H}}}$ is statistically equivalent with $\bm{{\rm H}}_a$, which does not hold when considering evanescent waves. Therefore, we will replace $\bm{{\rm H}}_a$ with $\bm{\widetilde{{\rm H}}}$ to derive the channel capacity in near-field communications.

As shown in Fig.~\ref{ellipse}, when considering evanescent waves, the angular response matrix can be expanded to include wave-numbers out of 2D lattice ellipse. In other words, sampling points $n_s$ and $n_r$ are much larger. We assume that the channels are independent identically distributed (i.i.d), which means the variances of matrix $\textbf{H}_a$'s elements naturally decouple:
\begin{equation}
    \sigma^2(l_x,l_y,m_x,m_y) = \sigma^2_s(m_x,m_y)\sigma^2_r(l_x,l_y)
\end{equation}
where $\sigma^2_s(m_x,m_y)$ and $\sigma^2_r(l_x,l_y)$ account for the power transfer at source and receiver, respectively.

Therefore, the matrix consisting of $n_rn_s$ variances can be denoted by:
\begin{equation}
    \bm{\Sigma} = vec\{\bm{\Sigma}_r\} vec\{\bm{\Sigma}_s\}^{\rm H}
    \label{sigma_decomposition}
\end{equation}
where $\bm{\Sigma}_r$ and $\bm{\Sigma}_s$ are variance matrices at source and receiver respectively.

To clarify the impact of evanescent waves, we assume $\bm{\Sigma}_r$ = $\bm{\Sigma}_{rin}$ + $\bm{\Sigma}_{rout}$, where $\bm{\Sigma}_{rin}$ and  $\bm{\Sigma}_{rout}$ represent variance matrices of sampling points in and out of the traditional support ellipse respectively. As shown in Lemma 1 of work \cite{b9}, angular random matrix can be obtained as
\begin{equation}
    \bm{{\rm H_a}} = \bm{\Sigma}\odot\bm{{\rm W}}
\end{equation}
where $\textbf{W}$ is the matrix with i.i.d. circularly-symmetric, complex-Gaussian random entries.

Thus, the capacity of near-field communications with evanescent waves can be calculated by:
\begin{equation}
    C = \rm \mathop{\max}_{ \textbf{Q}_a:tr(\textbf{Q}_a)<=1} \mathbb{E}\{log_2det( \textbf{I}_{n_r}+snr \bm{\widetilde{{\rm H}}} \textbf{Q}_a \bm{\widetilde{{\rm H}}^H})\}
    \label{near_field_capacity}
\end{equation}
where $\bm{\widetilde{{\rm H}}} = \rm \textbf{e}^{j\bm{{\rm\Gamma_r}}} \bm{{\rm H_a}} \textbf{e}^{-j\bm{{\rm\Gamma_s}}}.$

If the sampling points are in the traditional 2D ellipse region, $\gamma_r$ and $\gamma_s$ are real and $\bm{\widetilde{{\rm H}}}$ is statistically equivalent to $\bm{{\rm H_a}}.$ However, when the sampling points are outside the 2D ellipse region (i.e., evanescent waves are used for transmissions), $\gamma_r$ and $\gamma_s$ are imaginary and the signal power decays exponentially with the transmission distance (i.e.,  $\bm{\widetilde{{\rm H}}}$ is not statistically equivalent to $\bm{{\rm H_a}})$.

For simplicity, we assume: (1) instantaneous channel state information is only available at the receiver; (2) communications happen only  in interior 2D ellipse region or exterior 2D ellipse region (i.e., there is no cross region communications). Under assumption 1, the ergodic capacity in (\ref{near_field_capacity}) is achieved by an i.i.d input vector $\bm{x}_a$ with $\bm{{\rm Q}}_a = \frac{1}{n_s}\bm{{\rm I}}_{n_s}$ and channel capacity is given by
\begin{equation}
    C =  \sum_{i=1}^{{\rm rank}(\widetilde{\boldsymbol{{\rm H}}})}\mathbb{E}\{ {\rm log_2}(1+\frac{{\rm snr}}{n_s}\lambda_i(\widetilde{\boldsymbol{{\rm H}}}\widetilde{\boldsymbol{{\rm H}}}^{\rm H})) \}
    \label{near_field_capacity_average_power}
\end{equation}
where $\{\lambda_i(\bm{{\rm A}})\}$ are the eigenvalues of an arbitrary $\bm{{\rm A}}$.
Assuming the transmitting array is on the $x-y$ plane (i.e., $s_z$ = 0), then (\ref{near_field_capacity_average_power}) can be decomposed into:
\begin{equation}
    C = \sum_{i=1}^{{\rm rank}(\widetilde{\boldsymbol{{\rm H}}})}\mathbb{E}\{{\rm log_2}[1+\frac{{\rm snr}}{n_s}\lambda_i(e^{j\boldsymbol{\Gamma_r}}\boldsymbol{\Sigma}\odot \boldsymbol{{\rm W}}  {(e^{j\boldsymbol{\Gamma_r}}\boldsymbol{\Sigma}\odot \boldsymbol{{\rm W}} )}^{\rm H})] \}
    \label{capacity_decomposition}
\end{equation}
Applying (\ref{sigma_decomposition}) and $\bm{\Sigma}_r$ = $\bm{\Sigma}_{rin}$ + $\bm{\Sigma}_{rout}$, we can obtain:
\begin{equation*}
    \boldsymbol{\Sigma} =  [vec(\boldsymbol{\Sigma}_{rin})+vec(\boldsymbol{\Sigma}_{rout})][{vec(\boldsymbol{\Sigma}_{sin})}^{\rm H}+{vec(\boldsymbol{\Sigma}_{sout})}^{\rm H}]
\end{equation*}

Under assumption 2, we can simplify the above equation into:
\begin{equation}
    \boldsymbol{\Sigma} = vec(\boldsymbol{\Sigma}_{rin}){vec(\boldsymbol{\Sigma}_{sin})}^{\rm H}+vec(\boldsymbol{\Sigma}_{rout}){vec(\boldsymbol{\Sigma}_{sout})}^{\rm H}
    \label{sigma_decomposition_2}
\end{equation}
The capacity simulation results in Section \ref{simulation} are obtained through plugging (\ref{sigma_decomposition_2}) into (\ref{capacity_decomposition}).

\section{SIMULATION RESULTS}
\label{simulation}
Numerical results are performed to validate our theoretical results for a rectangular space $\mathcal{V}_3$.
Initially, we set simulation parameters as Table.~\ref{parameter_set}, then the DoF improvement percentage of near-field communications compared to the conventional far-field communication is computed through (\ref{theorem}) with one variable varying each time.

\begin{table}[htbp]
\caption{SIMULATION PARAMETERS}
\begin{center}
\begin{tabular}{|c|c|}
\hline
\multicolumn{2}{|c|}{\textbf{Simulation Parameters}} \\
\hline
\textbf{Antenna linear dimension D} & \textbf{Transmission distance} \\
\hline
$D=0.5$m& $z=0.5$m $^{\mathrm{*}}$\\
\hline
\textbf{Operation frequency}& \textbf{Power ratio} \\
\hline
$f=3G$Hz $^{\mathrm{*}}$& $P_{send}/P_{rece}=125.56$dB $^{\mathrm{*}}$\\
\hline
\multicolumn{2}{l}{The parameters with$^{\mathrm{*}}$ will be changed in the simulation stage}
\end{tabular}
\label{tab1}
\end{center}
\label{parameter_set}
\end{table}

In Fig.~\ref{distance}, it is shown that as the transmission distance becomes larger, the additional DoF provided by evanescent waves decreases rapidly. When the receiver is in the corresponding far-field region of the setting, the  DoF improvement can be neglected, which validates that the extra DoF gain only be obtained in the near-field region.
In Fig.~\ref{powerratio}, it it shown that as transmission power increases, extra DoF gain increases rapidly at first but converges to a constant mainly determined by the transmission distance, where the extra DoF gain can achieves more than 30\% improvement. This indicates that additional DoF gain is not mainly driven by transmit power consumption, but the EM nature of near-field communication.

\begin{figure}[htbp]
\centering
\includegraphics[width=0.35\textwidth]{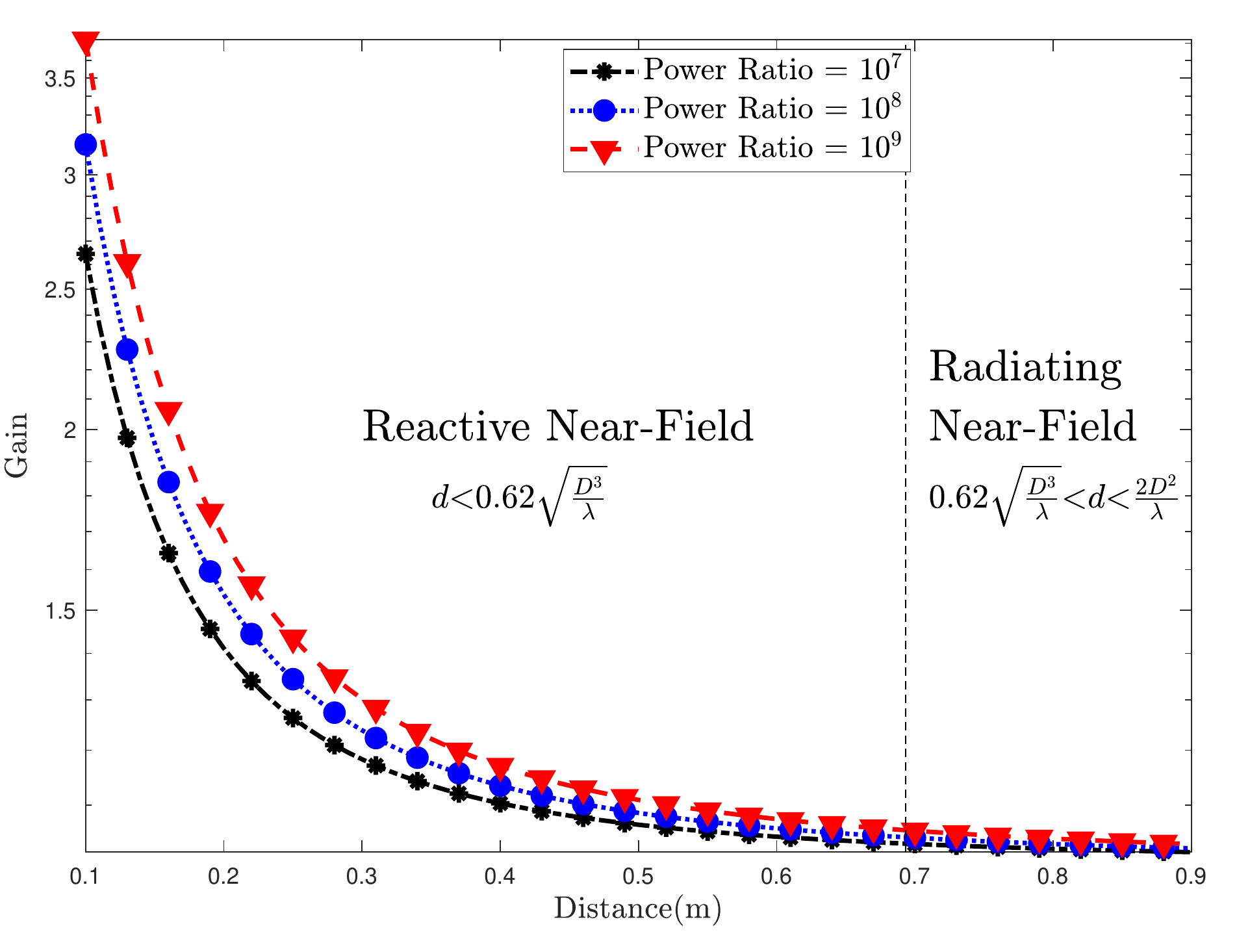}
\caption{Improved DoF percentage as a function of transmission distance.}
\label{distance}
\end{figure}


\begin{figure}[htbp]
\centering
\includegraphics[width=0.33\textwidth]{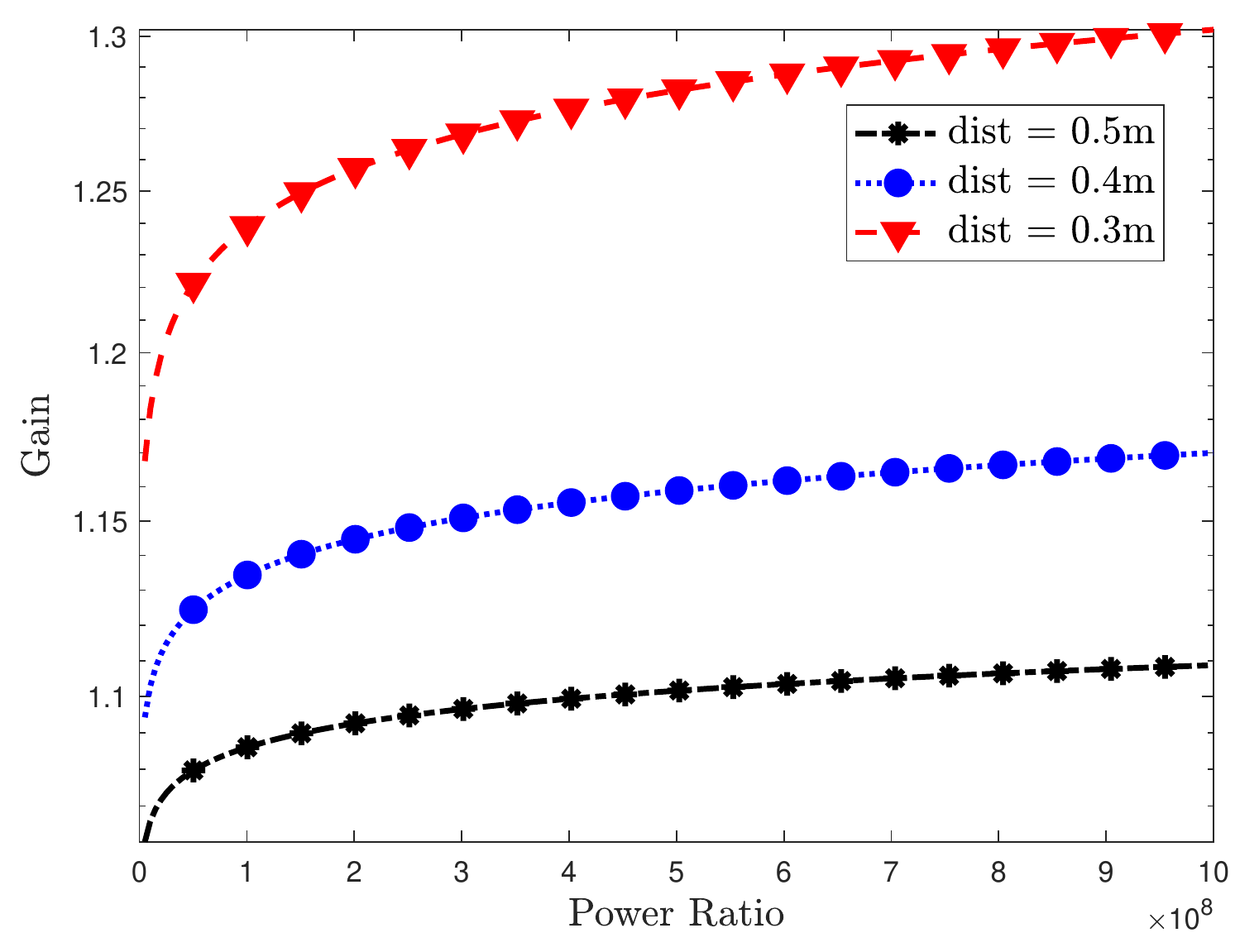}
\caption{Improved DoF percentage as a function of the ratio between transmission power and noise power.}
\label{powerratio}
\end{figure}

Next, we set simulation parameters as Table.~\ref{parameter_set_2}, then the channel capacity improvement of near-field communications is computed through (\ref{traditional_capacity}) and (\ref{capacity_decomposition}).

\begin{table}[htbp]
\caption{CAPACITY SIMULATION PARAMETERS}
\begin{center}
\begin{tabular}{|c|c|}
\hline
\multicolumn{2}{|c|}{\textbf{Simulation Parameters}} \\
\hline
\textbf{Antenna linear dimension D} & \textbf{Operation frequency} \\
\hline
$D=10\lambda$ & $f=300M/900M/3G$Hz\\
\hline
\multicolumn{2}{|c|}{\textbf{Transmission Distance}} \\
\hline
\multicolumn{2}{|c|} {$z=0.01 - 1$m $^{\mathrm{*}}$} \\
\hline
\multicolumn{2}{l}{The parameters with$^{\mathrm{*}}$ will be changed in the simulation stage}
\end{tabular}
\label{tab1}
\end{center}
\label{parameter_set_2}
\end{table}

In Fig.~\ref{capacity}, it is shown that as the transmission distance becomes larger, the additional channel capacity provided by evanescent waves decreases rapidly. Moreover, with the increase of the operating frequency, the extra capacity goes down faster.  When the receiver goes to the corresponding far-field region of the setting, the DoF improvement can be neglected, which further validates that extra DoF benefit can only be obtained in the near-field region.

\begin{figure}[htbp]
\centering
\includegraphics[width=0.35\textwidth]{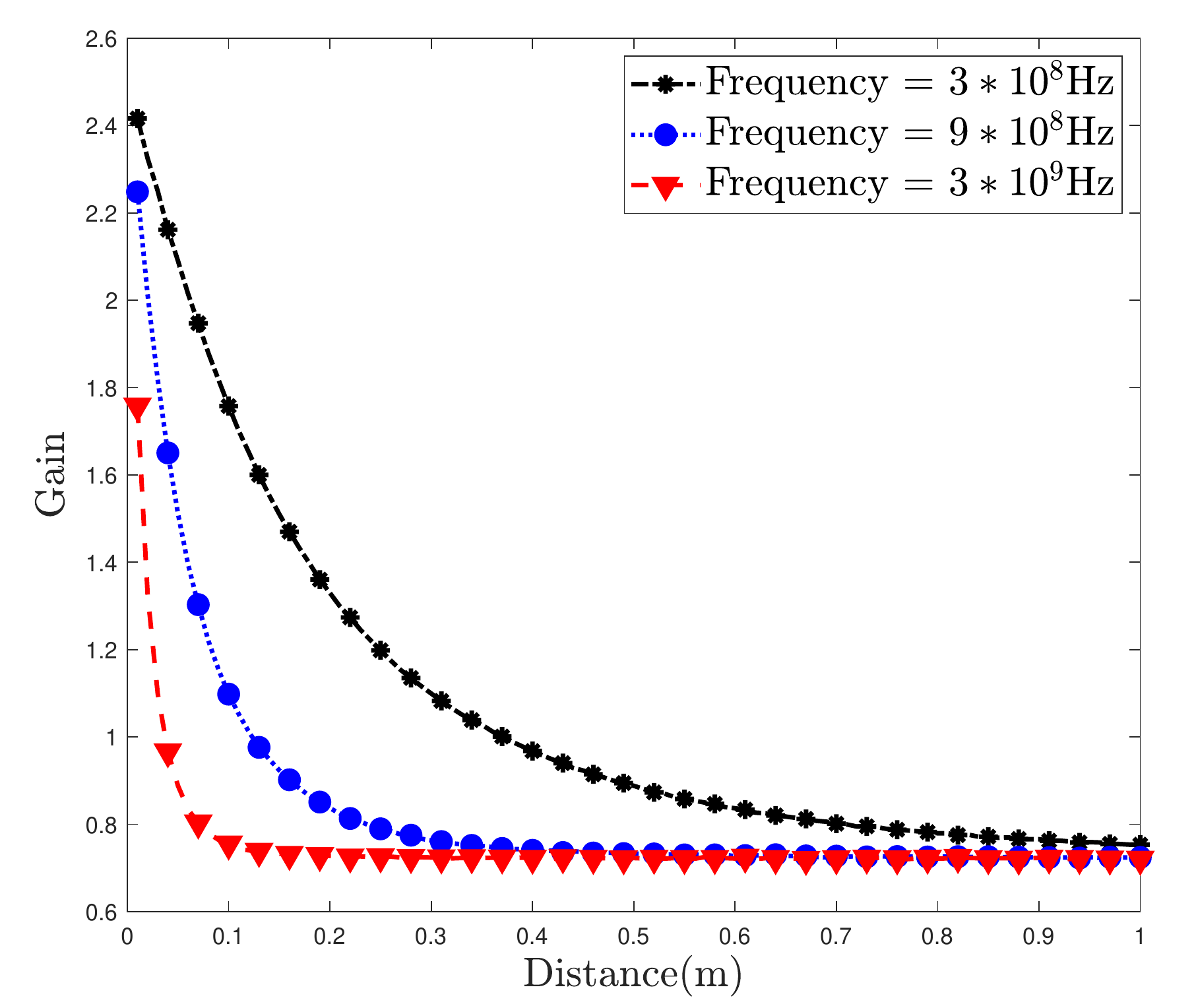}
\caption{Improved capacity percentage is a function of transmission distance.}
\label{capacity}
\end{figure}

%

\section{CONCLUSION}
Based on the Fourier plane-wave series representation of a spatially-stationary scattering channel, we derived a theoretical DoF and capacity gain in the near-field region by considering evanescent waves. This theoretical result captures the essence of EM propagation and also illuminates the possible benefits brought by evanescent waves.The analysis is limited to an isotropic scattering environment but can be extended to the non-isotropic case through the linear-system theoretic interpretation of plane-wave propagation. Numerical simulations demonstrates the validity of the derivations and its coherence with conventional DoF analysis in the far-field region.


\begin{thebibliography}{00}
\bibitem{b1} C. Huang, \emph{et al.}, ``Holographic MIMO surfaces for 6G wireless networks: opportunities, challenges, and trends'', \emph{IEEE Wireless Commun.}, vol. 27, no. 5, pp. 118-125, Oct. 2020.
\bibitem{b2} S. Hu, F. Rusek and O. Edfors, "Beyond Massive MIMO: The Potential of Data Transmission With Large Intelligent Surfaces," \emph{Trans. Signal Process.}, vol. 66, no. 10, pp. 2746-2758, 15 May, 2018.
\bibitem{b3} A. S. Y. Poon, R. W. Brodersen and D. N. C. Tse, "Degrees of freedom in multiple-antenna channels: a signal space approach," \emph{IEEE Trans. Inf. Theory}, vol. 51, no. 2, pp. 523-536, Feb. 2005.
\bibitem{b4} M. Franceschetti, "On Landau’s Eigenvalue Theorem and Information Cut-Sets," \emph{IEEE Trans. Inf. Theory}, vol. 61, no. 9, pp. 5042-5051, Sept. 2015, doi: 10.1109/TIT.2015.2456874.
\bibitem{b5} S. S. A. Yuan, Z. He, X. Chen, C. Huang and W. E. I. Sha, "Electromagnetic Effective Degree of Freedom of an MIMO System in Free Space," \emph{IEEE Ant. and Wireless Prop. Lett.}, vol. 21, no. 3, pp. 446-450, March 2022, doi: 10.1109/LAWP.2021.3135018.
\bibitem{b6} N. Decarli and D. Dardari, "Communication Modes With Large Intelligent Surfaces in the Near Field," \emph{IEEE Access}, vol. 9, pp. 165648-165666, 2021, doi: 10.1109/ACCESS.2021.3133707.
\bibitem{b7} Z. Zhang and L. Dai, "Pattern-Division Multiplexing for Continuous-Aperture MIMO," \emph{ICC 2022 - IEEE International Conference on Communications}, 2022, pp. 3287-3292, doi: 10.1109/ICC45855.2022.9839220.
\bibitem{b8} A. Pizzo, T. L. Marzetta and L. Sanguinetti, "Degrees of Freedom of Holographic MIMO Channels," \emph{Proc. IEEE SPAWC,} 2020, pp. 1-5.
\bibitem{b9} A. Pizzo, L. Sanguinetti and T. L. Marzetta, "Fourier Plane-Wave Series Expansion for Holographic MIMO Communications," \emph{IEEE Trans. Wireless Commun., }, doi: 10.1109/TWC.2022.3152965.
\bibitem{b10} A. Pizzo, T. L. Marzetta and L. Sanguinetti, "Spatially-Stationary Model for Holographic MIMO Small-Scale Fading," \emph{IEEE J. Sel. Areas Commun.,}, vol. 38, no. 9, pp. 1964-1979, Sept. 2020.
\bibitem{b11} L. Wei et al., "Multi-User Holographic MIMO Surfaces: Channel Modeling and Spectral Efficiency Analysis," \emph{IEEE J. Sel. Top. Sign. Process.} doi: 10.1109/JSTSP.2022.3176140.
\end{thebibliography}
\end{document}